\documentclass[journal, onecolumn, 12pt]{IEEEtran}
\usepackage{setspace}
\usepackage{amsfonts}
\usepackage{indentfirst}
\usepackage{amsmath,amsthm}
\usepackage{amssymb}
\usepackage{cite} 
\usepackage{color}
\usepackage{graphicx}
\usepackage{multirow,makecell,diagbox}
\usepackage[para]{threeparttable}
\newtheorem{Theorem}{Theorem}
\newtheorem{Corollary}{Corollary}
\newtheorem{Definition}{Definition}
\newtheorem{Proposition}{Proposition}
\newtheorem{Remark}{Remark}
\newtheorem{Lemma}{Lemma}
\newtheorem{Example}{Example}

\begin{document}

\title{New Complementary Sets with Low PAPR Property under Spectral Null Constraints}
\author{Yajing Zhou,
\thanks{Y. Zhou is with the School of Information Science and Technology, Southwest Jiaotong University,
Chengdu, 610031, China  (e-mail:
zhouyajing@my.swjtu.edu.cn).}
Yang Yang, \thanks{Y. Yang and Z. Zhou are with the School of Mathematics, Southwest Jiaotong University,
Chengdu, 610031, China  (e-mail:
yang\_data@swjtu.edu.cn, zzc@swjtu.edu.cn).}
Zhengchun Zhou, 
Kushal Anand, \thanks{K. Anand is with School of Electrical and Electronic Engineering, Nanyang Technological University, Singapore (email: KUSH0005@e.ntu.edu.sg).}
Su Hu, \thanks{S. Hu is with the National
Key Laboratory on Communications, University of Electronic Science and
Technology of China, Chengdu, China (E-mail: husu@uestc.edu.cn).}
and Yong Liang Guan \thanks{Y. L.  Guan is with School of Electrical and Electronic Engineering, Nanyang Technological University, Singapore (email: EYLGuan@ntu.edu.sg).}
}

\maketitle
\date{}
\begin{abstract}
Complementary set sequences (CSSs) are useful for dealing with the high peak-to-average power ratio (PAPR) problem in orthogonal frequency division multiplexing (OFDM) systems. In practical OFDM transmission, however, certain sub-carriers maybe reserved and/or prohibited to transmit signals, leading to the so-called \emph{spectral null constraint} (SNC) design problem. For example, the DC sub-carrier is reserved to avoid the offsets in D/A and A/D converter in the LTE systems. While most of the current research focus on the design of low PAPR CSSs to improve the code-rate, few works address the aforementioned SNC in their designs. This motivates us to investigate CSSs with SNC as well as low PAPR property. In this paper, we present systematic constructions of CSSs under SNCs and low PAPR. First, we show that mutually orthogonal complementary sets (MOCSs) can be used as \emph{seed sequences} to generate new CSSs with SNC and low PAPR, and then provide an iterative technique for the construction of MOCSs which can be further used to generate complementary sets (CSs) with low PAPRs and spectral nulls at \emph{varying} positions in the designed sequences. Next, inspired by a recent idea of Chen,  we propose a novel construction of these \emph{seed} MOCSs with non-power-of-two lengths from  generalized Boolean functions.

\end{abstract}

\section{Introduction}

Orthogonal frequency-division multiplexing (OFDM) is a modulation technique that transmits data by dividing it into several low rate data streams which are modulated over a number of sub-carriers \cite{00-Nee-Prasad}. Due to its robustness to intersymbol interference in multi-path fading and low complexity equalization at the receiver end, this technique has attracted much attention over the past decades, and has been widely adopted in several wireless system standards and products, such as IEEE 802.11, IEEE 802.16 \cite{09-IEEE} and so on. However, a major drawback of uncoded OFDM signals is that the transmitted signals can suffer from high peak-to-average power ratio (PAPR) which degrades the transmission power efficiency \cite{05-Goldsmith}.

While several signal processing techniques can be used to combat the PAPR issue \cite{00-Nee-Prasad},
an elegant approach to deal with this problem from the coding perspective is to construct codewords with low PAPR, and then apply them to a code-keying OFDM system. In this context, an important work was done by Davis and Jedwab \cite{99-Davis-Jedwab} by constructing codebooks from polyphase Golay complementary sequences (GCSs)\cite{61-Golay} using the algebraic tool of generalized Boolean functions. Therein, two properly chosen polyphase GCSs form a Golay complementary pair (GCP) with zero out-of-phase aperiodic autocorrelation sums, with each GCS producing an OFDM waveform with a PAPR upper bounded by 2 when this GCS is spread over the frequency domain \cite{91-Popovic}, \cite{14-Wang-Gong}.

Apart from their application in PAPR reduction, GCSs and GCPs have been used in many other scenarios such as Doppler resilient radar waveform design \cite{08-Pezeshki-Calderbank}, optimal channel estimation \cite{01-Spasojevic-Georghiades}, \cite{07-Wang-Abdi}, and interference-free multicarrier CDMA \cite{01-Chen-Yeh}, \cite{14-Liu-Guan-Para}, \cite{14-Liu-Guan-Chen}, etc. However, the lengths of these constructed GCPs and GCSs based on Boolean functions are all powers of two. Recently, Chen proposed novel constructions of $q$-ary (for even $q$) complementary sets (CSs) of non-power-of-two length based on generalized Boolean functions in \cite{16-ChenCY}, \cite{17-ChenCY} and \cite{18-ChenCY}. In \cite{17-ZLWang}, they proposed a novel construction
of $q$-ary CSs of non-power-of-two length based on the concatenation
of the sequences in GCPs or CSs. These classes of sequences enjoy low PAPR and good error correction capabilities, but their code-rates over phase-shift keying (PSK) signal constellations become very low with large code lengths.

Extension of polyphase GCSs to Quadrature Amplitude Modulation (QAM) GCSs was proposed in \cite{01-Rossing-Tarokh}. Owing to their larger set size, QAM GCSs exhibit higher code-rates, although with a slight (yet tolerable) increase of PAPR (compared to the polyphase GCSs). In \cite{01-Rossing-Tarokh}, R\"{o}{\ss}ing and Tarokh were the first to construct 16-QAM GCSs with maximum PAPR of 3.6 from the weighted sum of two quaternary GCSs. Subsequently, Chong et al. developed an algebraic construction of 16-QAM GCSs using generalized Boolean functions \cite{03-Chong-Venkataramani}. It was shown in \cite{03-Chong-Venkataramani} that an OFDM system with 16-QAM GCSs can indeed achieve a higher code-rate than that with only binary or quaternary GCSs, given the same PAPR constraint. Later, extension to the constructions of 64-QAM GCSs were reported in \cite{10-Chang-Li} and \cite{06-Lee-Golomb}. In \cite{08-LiY}, some corrections were provided for the sequence pairing descriptions of 16-QAM GCSs \cite{03-Chong-Venkataramani} and 64-QAM GCSs \cite{06-Lee-Golomb}. Furthermore, generalized Case I-III $4^q$-QAM ($q\geq1$) GCSs were investigated in \cite{10-LiY} and Generalized Case IV-V $4^q$-QAM ($q\geq3$) GCSs using selected Gaussian integer pairs were reported in \cite{13-Liu-Li-Guan}.

While the above mentioned works focus on complementary set sequence (CSS) designs with low PAPR property, they do not consider an important practical constraint, i.e., \emph{spectral nulls} in OFDM transmission, as explained next.
In OFDM systems, certain sub-carriers are reserved and are prohibited to transmit signals \cite{IEEE}. For example, the DC sub-carrier is reserved, i.e., ``spectrally nulled", to avoid the offsets in the D/A and A/D converters in RF transmission. Similarly, the guard bands at the spectrum edges are also nulled to prevent interferences to the adjacent sub-carriers \cite{09-IEEE,11-Hamilton-Ma}. Furthermore, it has been reported in the literature that adjusting the positions of the null sub-carriers within the OFDM symbol is beneficial for applications such as accurate CFO estimation \cite{01-Ghogho-Swami,12-Wang-Ho,11-Tsai-2}. More recently, the need for OFDM sequences with spectral null constraint (SNC) (also referred to as \emph{non-contiguous} OFDM \cite{12-Sakran-Nasr}) is motivated by their possible use cases in the Cognitive Radio (CR) communications \cite{09-Mahmoud-Yucek,14-Hu-Liu,18-Liu-Guan}. In OFDM based transmission for CR, the secondary users (SUs) are allowed to transmit only on those sub-carriers which are not occupied by the primary user (PU), thus requiring spectral nulls at specific positions (corresponding to the sub-carriers occupied by the PU) in the SUs' transmitted codewords \cite{15-Ni-Jiang}. Unfortunately, due to the explicit constraints of the spectral nulls in the above application scenarios, the \emph{traditional} CSs (which do not consider SNC in their designs) may not be applicable anymore. Moreover, in the non-contiguous OFDM waveforms for CR, PAPR problem is even worse than the normal OFDM \cite{12-Sakran-Nasr}.  Thus, an intriguing question that follows is how to construct new complementary sets containing sequences which can be used as low PAPR OFDM codewords with spectral nulls at specific positions.

This work attempts to address the above problem by providing a systematic construction of CSs with low PAPR, yet satisfying the spectral constraints for some specific scenarios. In particular, we start by an iterative method of generating mutually orthogonal CSs (MOCSs) and prove that low PAPR CSs with spectral nulls at certain positions can be generated using MOCSs, and then we go on to construct MOCSs with non-power-of-two lengths (in contrast to the very specific power-of-two length sequences) using generalized Boolean functions which is inspired by a recent idea of Chen \cite{16-ChenCY}. The MOCSs designed using the proposed iterative techniques can, in fact, be used to construct low PAPR CSs with spectral nulls at more general positions in the sequences which may be useful in applications mentioned in \cite{01-Ghogho-Swami,12-Wang-Ho}.

The remainder of this paper is organized as follows. Section \ref{sec-pre} gives the preliminaries and the mathematical tools used in the paper. In Section III, we present an iterative method to construct MOCSs under SNC with bounded PAPR from the known sequences. In Section IV,  we introduce a new construction of MOCSs with non-power-of-two lengths based on generalized Boolean functions. In Section V we calculate the code-rate of the codebook of the proposed construction in Section IV. Finally,  Section VI concludes this paper with some remarks.

\section{ Preliminaries}\label{sec-pre}

\subsection{Aperiodic Correlations and Complementarity}

Let $\mathbf{a}=(a(0),a(1),\cdots,a(L-1))$ and $\mathbf{b}=(b(0),b(1),\cdots,b(L-1))$ be two complex-valued sequences of length $L$.  The aperiodic cross-correlation between $\mathbf{a}$ and $\mathbf{b}$ at a time shift $u$ is defined by
\begin{eqnarray*}
 R_{\mathbf{a,b}}(u) &=&\left\{\begin{array}{ll}
\sum_{i=0}^{L-1-u}a(i)b^*(i+u), & 0\leq u\leq L-1; \\
\sum_{i=0}^{L-1+u}a(i-u)b^*(i), &-(L-1)\leq u\leq-1; \\
 0, & |u|\geq L.
   \end{array}\right.
\end{eqnarray*}
It is easily verified that
\begin{eqnarray}\label{eq-u}
R_{\mathbf{a,b}}(u)=R^*_{\mathbf{b,a}}(-u).
\end{eqnarray}
If $\mathbf{a=b}$, $R_{\mathbf{a,b}}(u)$ reduces to the aperiodic auto-correlation of $\mathbf{a}$ and will be written as $R_\mathbf{a}(u)$ for simplicity.

\begin{Definition}[\cite{72-Tseng-Liu}]
Let ${\mathcal{A}}=\{\mathbf{a}_{i}\}_{i=1}^{N}$ be a set of $N$ complex-valued sequences of length $L$.  It is said to be a complementary set (CS) of size $N$ if $\sum_{i=1}^NR_{\mathbf{a}_{i}}(u)=0$ for any  $u\neq0$, and any sequence in this set is called a complementary set sequence (CSS). In particular, when $N=2$, the set is called a Golay complementary pair (GCP), and any sequence in this pair is called a Golay complementary sequence (GCS), or complementary sequence.
\end{Definition}

\begin{Definition}
Two sequence sets
${\mathcal{S}}_1=\{\mathbf{s}_{1,1},\mathbf{s}_{1,2},
  \cdots,\mathbf{s}_{1,N}\}$
and
${\mathcal{S}}_2=\{\mathbf{s}_{2,1},\mathbf{s}_{2,2},
  \cdots,\mathbf{s}_{2,N}\}$
are said to be mutually orthogonal if
\begin{eqnarray*}
\sum_{j=1}^{N}R_{\mathbf{s}_{1,j},\mathbf{s}_{2,j}}(u) =0, && \textrm{for~~all~}~~-L+1\le u\le L-1.
\end{eqnarray*}
\end{Definition}

\begin{Definition}
Let ${\mathcal{S}}=\{{\mathcal{S}}_1,{\mathcal{S}}_2,\cdots,{\mathcal{S}}_M\}$, where each ${\mathcal{S}}_c~(1\leq c\leq M)$ is a complementary set consisting of $N$ length-$L$ sequences.  ${\mathcal{S}}$ is called an $(M,N,L)$-MOCS (mutually orthogonal complementary set) if
\begin{eqnarray*}
\sum_{j=1}^NR_{\mathbf{s}_{i,j},\mathbf{s}_{k,j}}(u) =0, ~~\forall~~1\le i\ne k\le M  ~\textrm{and}~ -L+1\le u\le L-1.
\end{eqnarray*}
{
\begin{Remark}\label{rem-han}
It turns out in Theorem 4 of \cite{11-Han} that $M\leq N$ for any $(M,N,L)$-MOCS.
An MOCS with $M=N$ is also called a complete complementary code (CCC). Thanks to a systematic framework in \cite{11-Han} (see Theorems 5 and 7 in \cite{11-Han}), an $(N,N,L)$-MOCS can be obtained from a unitary-like matrix of order $N$, where $L$ is an arbitrary number whose factors are not greater than $N$.
\end{Remark}}
\end{Definition}


\subsection{Generalized Boolean Functions}

Let $q$ be a positive integer and $\mathbb{Z}_q=\{0,1,\cdots,q-1\}$ denote the set of integers modulo $q$.
For $\mathbf{\underline{x}}=(x_1,x_2,\cdots,x_m)\in\mathbb{Z}_2^m$, a generalized Boolean function $f(\mathbf{\underline{x}})$ is defined as a mapping $f$ from $\{0,1\}^m$ to $\mathbb{Z}_q$.
Given $f(\mathbf{\underline{x}})$, define
\begin{eqnarray}\label{eq-GBF}
  \mathbf{f}=(f(\underline{0}),f(\underline{1}),...,
  f(\underline{2^{m}-1})),
\end{eqnarray}
where
$f(\underline{i})=f(i_1,i_2,\cdots,i_{m})$, and $(i_1,i_2,\cdots,i_{m})$ is the binary representation of $i=\sum_{k=1}^mi_k2^{k-1}$ with $i_{m}$ denoting the most significant bit.

In this paper, we consider the  truncated version of the sequence $\mathbf{f}$ above. Specifically,  let  $\mathbf{f}^{(L)}$ be a sequence of length $L$ obtained from $\mathbf{f}$ by ignoring the last $2^m-L$ elements of the sequence $\mathbf{f}$. That is, $\mathbf{f}^{(L)}=(f(\underline{0}),f(\underline{1}),
\cdots,f(\underline{L-1}))$ is a sequence of length $L$.  Let $\xi=\exp(2\pi\sqrt{-1}/q)$ be a primitive $q$-th complex root of unity. One can naturally associate a complex-valued sequence $\psi(\mathbf{f}^{(L)})$ of length $L$ with $\mathbf{f}^{(L)}$ as
\begin{eqnarray}\label{eq-fL}
  \psi(\mathbf{f}^{(L)}) &:=&(\xi^{f(\underline{0})},\xi^{f(\underline{1})},
  \cdots,\xi^{f(\underline{L-1})}).
\end{eqnarray}
From now on, whenever the context is clear, we ignore the superscript of $\mathbf{f}^{(L)}$ unless the sequence length is specified.

\subsection{PAPR of OFDM Symbol}

Let us consider an $L$-sub-carrier OFDM system. The transmitted OFDM signal is the real part of the complex envelope, which can be written as
\begin{eqnarray}\label{st}
S_\mathbf{a}(t)= \sum_{i=0}^{L-1}a(i)e^{2\pi(f_{c}+i{\Delta}f)t\sqrt{-1}},~ 0\leq t<T,
\end{eqnarray}
where $f_{c}$ denotes the carrier frequency and ${\Delta}f=\frac{1}{T}$ denotes the sub-carrier spacing, with $T$ being the OFDM symbol duration. The sequence $\mathbf{a}=(a(i))$ of length $L$ is called the modulating sequence of the OFDM symbol.

The instantaneous power of an OFDM sequence (codeword) is given by $P_\mathbf{a}(t)=|S_\mathbf{a}(t)|^{2}$.
The PAPR of the OFDM sequence is then defined as:
\begin{eqnarray}
\hbox{PAPR}(\mathbf{a}) =\frac{\underset{t\in[0,T)}{\sup} P_{\mathbf{a}}(t)}{P_{av}({\mathbf{a}})},
\end{eqnarray}
where $P_{av}({\mathbf{a}})$ is the average power of $\mathbf{a}$.
Accordingly, the PAPR of a sequence set ${\mathcal{A}}=\{\mathbf{a}_1,\mathbf{a}_2,\cdots,
\mathbf{a}_N\}$ is  defined as
\begin{eqnarray*}
 \hbox{PAPR}({\mathcal{A}}) &=& \underset{\mathbf{a}_i\in {\mathcal{A}}}{\max}~\hbox{PAPR}(\mathbf{a}_i).
\end{eqnarray*}

Define the code-rate of a code-keying OFDM as $R(\mathcal{C}):=\frac{\log_q|\mathcal{C}|}{L}$, where $q$ is the constellation size, $|\mathcal{C}|$ and $L$ denote the set size of codebook $\mathcal{C}$ and the codeword length (or the number of sub-carriers) respectively.

\section{An Iterative Construction of CSs under Spectral Null Constraints}

In this section, we first recall an upper bound on the PAPR of CS sequences, and then present an iterative method to construct CSs under spectral null constraints (SNCs) with low PAPR property.

\subsection{An Upper Bound on PAPR of CSs}
The following bound due to Liu and Guan will be used to estimate PAPR of
the proposed CSs in the sequel.

\begin{Lemma}[Lemma 2 of \cite{16-Liu-Guan}]\label{thm_1} 
Let ${\mathcal{A}}$ be a CS of size $N$ in which all the sequences have the same length and energy. Then the PAPR of ${\mathcal{A}}$ is upper bounded by $N$.
\end{Lemma}

\subsection{An Iterative Construction of CSs under SNCs}
%
%
%
%
%

In this part, we provide an iterative method to construct CSs under SNCs with low PAPRs from MOCSs, and the position of the spectrum null can be not only in the middle but also symmetrical.
\begin{enumerate}
\item[Step 1:] Select an $(M,N,L)$-MOCS: $\mathcal{W}=\{\mathcal{W}_1,\mathcal{W}_2,
     \cdots,\mathcal{W}_M\}$, where $\mathcal{W}_i=\{\mathbf{w}_{i,j}\}_{j=1}^N$ and
     \begin{eqnarray*}
      \mathbf{w}_{i,j}=\{w_{i,j}(0),w_{i,j}(1),
      \cdots,w_{i,j}(L-1)\}, && 1\leq i\leq M,~1\leq j\leq N.
     \end{eqnarray*}
\item[Step 2:]
Let $M_1=\lfloor\frac{M}{2}\rfloor$, $b_1$ be any given non-negative integer, and $L_1=2L+b_1$. Obtain the following sequence set
\begin{eqnarray*}
 \mathcal{W}^{(1)}_i=\{\mathbf{w}^{(1)}_{i,j}\}_{j=1}^N, ~~~1\leq i\leq M_1
\end{eqnarray*}
from $\mathcal{W}_i$ and $\mathcal{W}_{i+M_1}$, where
\begin{eqnarray*}
  \mathbf{w}^{(1)}_{i,j}=(w^{(1)}_{i,j}(0), w^{(1)}_{i,j}(1),\cdots,w^{(1)}_{i,j}(L_{1}-1))
\end{eqnarray*}
and
\begin{eqnarray}\label{eq-iter}
w^{(1)}_{i,j}(t)&=& \left\{ \begin{array}{ll}
w_{i,j}(t), & t=0,1,\cdots,L-1;\\
0,&t=L,\cdots,L+b_1-1;\\
w_{i+M_{1},j}(t-(L+b_1)),& t=L+b_1,\cdots,2L+b_1-1.
\end{array}\right.
\end{eqnarray}
~~~~~~~~~~~~~~~~~~~~~~~~~~~~~~~~~~~~~~~~~~~~~~~~~\vdots

\item[Step $k$:]
Let $M_{k-1}=\lfloor\frac{M}{2^{k-1}}\rfloor\geq1$ where $k\geq3$, $b_{k-1}$ be any given non-negative integer,  and $L_{k-1}=2L_{k-2}+b_{k-1}$. Generate a sequence set
$$
\mathcal{W}^{(k-1)}_i=\{\mathbf{w}^{(k-1)}_{i,j}\}_{j=1}^N, ~~1\leq i\leq M_{k-1}
$$
of sequence length $L_{k-1}=2L_{k-2}+b_{k-1}$ from  $\mathcal{W}^{(k-2)}_i$ and $\mathcal{W}^{(k-2)}_{i+M_{k-1}}$ as in (\ref{eq-iter}).
\end{enumerate}

\begin{Theorem}\label{MO}
The set $\mathcal{W}^{(1)}=\{\mathcal{W}^{(1)}_1,\mathcal{W}^{(1)}_2
,\cdots,\mathcal{W}^{(1)}_{M_1}\}$ is an $(M_1,N,2L+b_1)$-MOCS.
\end{Theorem}

\begin{proof}The proof of this theorem is divided into two steps. In the first step, we show that $\mathcal{W}^{(1)}_i~(1\leq i \leq M_1)$ is a CS of length $2L+b_1$. In the second one,  we prove that for $i\neq r$, $\mathcal{W}^{(1)}_i$ and $\mathcal{W}^{(1)}_{r}$ are mutually orthogonal.

\emph{Step 1:} We only need to prove that $\sum_{j=1}^{N}R_{\mathbf{w}^{(1)}_{i,j}}(u)=0$ holds for  $1\leq u\leq2L-1+b_1$.
We distinguish among the following three cases to achieve this goal.
\begin{itemize}
\item For $b_1=0$, from the definition of aperiodic cross-correlation, we have
      \begin{eqnarray}
        R_{\mathbf{w}^{(1)}_{i,j}}(u) &=& \left\{\begin{array}{ll}
R_{\mathbf{w}_{i,j}}(u)+R_{\mathbf{w}_{i+M_1,j},\mathbf{w}_{i,j}}^*
(L-u)+R_{\mathbf{w}_{i+M_1,j}}(u), &0< u\leq L-1; \\
R_{\mathbf{w}_{i,j},\mathbf{w}_{i+M_1,j}}(u-L), & L\leq u\leq 2L-1.\end{array}\right.
      \end{eqnarray}
Note that $\mathcal{W}_i$ and $\mathcal{W}_{i+M_1}$ are CSs, one has
\begin{eqnarray}
\sum_{j=1}^NR_{\mathbf{w}^{(1)}_{i,j}}(u) &=& \left\{\begin{array}{ll}
\sum_{j=1}^NR_{\mathbf{w}_{i+M_1,j},\mathbf{w}_{i,j}}^*
(L-u), &0< u\leq L-1; \\
\sum_{j=1}^NR_{\mathbf{w}_{i,j},\mathbf{w}_{i+M_1,j}}(u-L), & L\leq u\leq 2L-1.\end{array}\right.
     \end{eqnarray}

\item For $1\leq b_1\leq L-1$, we have
\begin{eqnarray}
R_{\mathbf{w}^{(1)}_{i,j}}(u) &=& \left\{\begin{array}{ll}
   R_{\mathbf{w}_{i,j}}(u)+R_{\mathbf{w}_{i+M_1,j}}(u), & 0< u\leq b_1; \\
R_{\mathbf{w}_{i,j}}(u)+R_{\mathbf{w}_{i+M_1,j},\mathbf{w}_{i,j}}^*
(L+b_1-u)+R_{\mathbf{w}_{i+M_1,j}}(u), &b_1< u\leq L-1; \\
 R_{\mathbf{w}_{i+M_1,j},\mathbf{w}_{i,j}}^*(L+b_1-u), & L-1<u\leq L+b_1-1;\\
R_{\mathbf{w}_{i,j},\mathbf{w}_{i+M_1,j}}(u-L-b_1), & L+b_1\leq u\leq 2L+b_1-1.\end{array}\right.
      \end{eqnarray}
This, together with the fact that $\mathcal{W}_i$ and $\mathcal{W}_{i+M_1}$ are CSs, implies that
\begin{eqnarray}
\sum_{j=1}^NR_{\mathbf{w}^{(1)}_{i,j}}(u) &=& \left\{\begin{array}{ll}
   0, & 0< u\leq b_1; \\
\sum_{j=1}^NR_{\mathbf{w}_{i+M_1,j},\mathbf{w}_{i,j}}^*
(L+b_1-u), &b_1< u\leq L-1; \\
 \sum_{j=1}^NR_{\mathbf{w}_{i+M_1,j},\mathbf{w}_{i,j}}^*(L+b_1-u), & L-1<u\leq L+b_1-1;\\
\sum_{j=1}^NR_{\mathbf{w}_{i,j},\mathbf{w}_{i+M_1,j}}(u-L-b_1), & L+b_1\leq u\leq 2L+b_1-1.\end{array}\right.
     \end{eqnarray}

\item For $b_1\geq L$, we have
      \begin{eqnarray}
        R_{\mathbf{w}^{(1)}_{i,j}}(u) &=& \left\{\begin{array}{ll}
   R_{\mathbf{w}_{i,j}}(u)+R_{\mathbf{w}_{i+M_1,j}}(u), & 0< u\leq L-1; \\
   0,&L\leq u\leq b_1;\\
 R_{\mathbf{w}_{i+M_1,j},\mathbf{w}_{i,j}}^*(L+b_1-u), & b_1+1\leq u\leq L+b_1-1;\\
R_{\mathbf{w}_{i,j},\mathbf{w}_{i+M_1,j}}(u-L-b_1), & L+b_1\leq u\leq 2L+b_1-1.\end{array}\right.
\end{eqnarray}
This leads to
\begin{eqnarray}
\sum_{j=1}^NR_{\mathbf{w}^{(1)}_{i,j}}(u) &=& \left\{\begin{array}{ll}
   0, & 0< u\leq L-1; \\
   0, &L\leq u\leq b_1;\\
 \sum_{j=1}^NR_{\mathbf{w}_{i+M_1,j},\mathbf{w}_{i,j}}^*(L+b_1-u), & b_1+1\leq u\leq L+b_1-1;\\
\sum_{j=1}^NR_{\mathbf{w}_{i,j},\mathbf{w}_{i+M_1,j}}(u-L-b_1), &  L+b_1\leq u\leq 2L+b_1-1.\end{array}\right.
      \end{eqnarray}
\end{itemize}
Combining the cases above and noting that $\mathcal{W}_i$ and $\mathcal{W}_{i+M_1}$ are mutually orthogonal, we have
\begin{eqnarray*}
 \sum_{j=1}^NR_{\mathbf{w}^{(1)}_{i,j}}(u)=0 ~~\textrm{for~~all~}~ 1\leq u\leq2L-1+b_1.
\end{eqnarray*}

\emph{Step 2:}
\begin{itemize}
\item For $ b_1=0$, from the definition of aperiodic cross-correlation, we have
      \begin{eqnarray}
        R_{\mathbf{w}^{(1)}_{i,j},\mathbf{w}^{(1)}_{r,j}}(u) = \left\{\begin{array}{ll}
R_{\mathbf{w}_{i,j},\mathbf{w}_{r,j}}(u)
   +R^*_{\mathbf{w}_{r+M_{1},j},\mathbf{w}_{i,j}}
   (L-u)&\\+R_{\mathbf{w}_{i+M_{1},j},
   \mathbf{w}_{r+M_{1},j}}(u), &0< u\leq L-1; \\
R_{\mathbf{w}_{i,j},\mathbf{w}_{r+M_{1},j}}
   (u-L), &L\leq u\leq2L-1;\\
   R^*_{\mathbf{w}_{r,j},\mathbf{w}_{i,j}}(-u)
   +R_{\mathbf{w}_{i+M_{1},j},\mathbf{w}_{r,j}
   (L+u)}&\\+R^*_{\mathbf{w}_{r+M_{1},j},
   \mathbf{w}_{i+M_{1},j}}(-u), & -L+1\leq u<0; \\
   R^*_{\mathbf{w}_{r,j},\mathbf{w}_{i+M_{1},j}}
   (-u-L), &-2L+1\leq u\leq -L.\\
   \end{array}\right.
      \end{eqnarray}

Since for any $ 1\leq s\ne h\leq N$, $\mathcal{W}_s$ and $\mathcal{W}_h$ are mutually orthogonal, then
\begin{eqnarray}\label{eq-w1}
\sum_{j=1}^NR_{\mathbf{w}^{(1)}_{i,j},\mathbf{w}^{(1)}_{r,j}}(u) =0, &&\textrm{for~~all~}~1-b_1-2L\leq u \leq 2L+b_1-1.
     \end{eqnarray}
  \item For $1\leq b_1\leq L-1$, from the definition of aperiodic cross-correlation, we have
      \begin{eqnarray}\nonumber
        &&R_{\mathbf{w}^{(1)}_{i,j},\mathbf{w}^{(1)}_{r,j}}(u) \\
        &&=\left\{\begin{array}{ll}
   R_{\mathbf{w}_{i,j},\mathbf{w}_{r,j}}(u)+
   R_{\mathbf{w}_{i+M_{1},j},\mathbf{w}_{r+M_{1},j}}(u), & 0\leq u\leq b_1; \\
R_{\mathbf{w}_{i,j},\mathbf{w}_{r,j}}(u)+R^*_{\mathbf{w}_{r+M_{1},j},\mathbf{w}_{i,j}}
   (v)+R_{\mathbf{w}_{i+M_{1},j},\mathbf{w}_
   {r+M_{1},j}}(u), &b_1< u\leq L-1; \\
 R^*_{\mathbf{w}_{r+M_{1},j},\mathbf{w}_{i,j}}(v), & L-1<u\leq L+b_1-1;\\
R_{\mathbf{w}_{i,j},\mathbf{w}_{r+M_{1},j}}
   (u-L-b_1), & L+b_1\leq u\leq2L+b_1-1;\\
   R^*_{\mathbf{w}_{r,j},\mathbf{w}_{i,j}}(-u)+
   R^*_{\mathbf{w}_{r+M_{1},j},\mathbf{w}_{i+M_{1},j}}(-u), & -b_1\leq u< 0; \\
R^*_{\mathbf{w}_{r,j},\mathbf{w}_{i,j}}(-u)
   +R_{\mathbf{w}_{i+M_{1},j},\mathbf{w}_{r,j}}
   (\bar{v})+R^*_{\mathbf{w}_{r+M_{1},j},\mathbf{w}_{i+M_{1},j}}(-u), &1-L\leq u< -b_1; \\
 R_{\mathbf{w}_{i+M_{1},j},\mathbf{w}_{r,j}}(\bar{v}), & 1-L-b_1\leq u<1-L;\\
R^*_{\mathbf{w}_{r,j},\mathbf{w}_{i+M_{1},j}}
   (-u-L-b_1), & 1-2L-b_1\leq u\leq -L-b_1.\end{array}\right.
      \end{eqnarray}
   where $v=L+b_1-u$ and $\bar{v}=L+b_1+u$.

Similar to Case 1, (\ref{eq-w1}) holds for $1\leq b_1\leq L-1$.
\item  For $b_1\geq L$, from the definition of aperiodic cross-correlation, we have
      \begin{eqnarray}
        R_{\mathbf{w}^{(1)}_{i,j},\mathbf{w}^{(1)}_{r,j}}(u) &=& \left\{\begin{array}{ll}
  R_{\mathbf{w}_{i,j},\mathbf{w}_{r,j}}(u)+R_{\mathbf{w}_{i+M_{1},j},\mathbf{w}_{r
   +M_{1},j}}(u), & 0\leq u\leq L-1; \\
   0, & L\leq u\leq b_1;\\
 R^*_{\mathbf{w}_{r+M_{1},j},\mathbf{w}_{i,j}}(L+b_1-u), & b_1+1\leq u\leq L+b_1-1;\\
R_{\mathbf{w}_{i,j},\mathbf{w}_{r+M_{1},j}}(u-L-b_1), & L+b_1\leq u\leq2L+b_1-1;\\
    R^*_{\mathbf{w}_{r,j},\mathbf{w}_{i,j}}(-u)+
   R^*_{\mathbf{w}_{r+M_{1},j},\mathbf{w}_{i+M_{1},j}}(-u), &1-L\leq u<0; \\
   0, & -b_1\leq u\leq -L;\\
 R_{\mathbf{w}_{i+M_{1},j},\mathbf{w}_{r,j}}(L+b_1+u), & 1-L-b_1\leq u\leq -1-b_1;\\
R^*_{\mathbf{w}_{r,j},\mathbf{w}_{i+M_{1},j}}
   (-u-L-b_1), & 1-2L-b_1\leq u\leq-L-b_1.\end{array}\right.
      \end{eqnarray}
Similar to Case 1, (\ref{eq-w1}) holds when $b_1\geq L$.
\end{itemize}

Combining the cases above, it can be concluded that for $i\neq r$, $\mathcal{W}^{(1)}_i$ and $\mathcal{W}^{(1)}_{r}$ are mutually orthogonal.
\end{proof}

The following corollary is a direct result of Theorem \ref{MO}.
\begin{Corollary}\label{wk}
For any $k\geq3$ and $M\geq 2^k$, $\mathcal{W}^{(k-1)}=\{\mathcal{W}^{(k-1)}_1,\mathcal{W}^{(k-1)}_2
,\cdots,\mathcal{W}^{(k-1)}_{M_{k-1}}\}$ is an $(M_{k-1},N,2^{k-1}L+\sum_{i=1}^{k-1}2^{{k-1}-i}b_i)$-MOCS.
\end{Corollary}

\begin{Remark}
  Note that when $M_{k-1}=1$, $\mathcal{W}^{(k-1)}$ in Corollary \ref{wk} is reduced to a CS of size $N$ with length $2^{k-1}L+\sum_{i=1}^{k-1}2^{k-1-i}b_i$. By Step 2 in the iterative construction we obtained  CSs with spectral nulls (zeros) only at the center of the sequences.
However, after $(k-1)$-th iteration we can get CSs (by Theorem~\ref{wk}) with spectral nulls at various positions (other than the center) within the sequence.
\end{Remark}

\begin{Example}
Let $\mathcal{W}=\{\mathcal{W}_1,\mathcal{W}_2,\mathcal{W}_3,
\mathcal{W}_4\}$ be a $(4,4,4)$-MOCS, where
\begin{eqnarray*}
      &&\mathcal{W}_1 =\{(1,1,1,1),(1,1,-1,-1),(-1,1,-1,1),(-1,1,1,-1)\}, \\
      && \mathcal{W}_2=\{(-1,1,-1,1),(-1,1,1,-1),(1,1,1,1),(1,1,-1,-1)\},  \\
      && \mathcal{W}_3 =\{(-1,-1,1,1),(-1,-1,-1,-1),(1,-1,-1,1),(1,-1,1,-1)\},   \\
      && \mathcal{W}_4 =\{(1,-1,-1,1),(1,-1,1,-1),(-1,-1,1,1),(-1,-1,-1,-1)\}.
    \end{eqnarray*}
     Let $b_1=1$ and
    \begin{eqnarray*}
      \mathcal{W}_1^{(1)}&=& \{(1,1,1,1,0,-1,-1,1,1),(1,1,-1,-1,0,-1,-1,-1,-1), \\
      &&(-1,1,-1,1,0,1,-1,-1,1),(-1,1,1,-1,0,1,-1,1,-1)\},\\
      \mathcal{W}_2^{(1)}&=&\{(-1,1,-1,1,0,1,-1,-1,1),(-1,1,1,-1,0,1,-1,1,-1),\\
      &&(1,1,1,1,0,-1,-1,1,1),(1,1,-1,-1,0,-1,-1,-1,-1)\}.
    \end{eqnarray*}
   By Theorem \ref{MO}, $\mathcal{W}^{(1)}=\{\mathcal{W}_1^{(1)},
   \mathcal{W}_2^{(1)}\}$ is a $(2,4,9)$-MOCS, and the sequence set
\begin{eqnarray*}
      &&\begin{array}{l}
    \{(1,1,1,1,0,-1,-1,1,1,0,0,-1,1,-1,1,0,1,-1,-1,1), \\
   (1,1,-1,-1,0,-1,-1,-1,-1,0,0,-1,1,1,-1,0,1,-1,1,-1), \\
   (-1,1,-1,1,0,1,-1,-1,1,0,0,1,1,1,1,0,-1,-1,1,1), \\
   (-1,1,1,-1,0,1,-1,1,-1,0,0,1,1,-1,-1,0,-1,-1,-1,-1) \}
          \end{array}
    \end{eqnarray*}
 is a CS of length 20.
\end{Example}

{
\begin{Remark}
The above construction is generic in the sense that it works for any  mutually orthogonal CSs.
{According to Theorem \ref{MO} and Corollary \ref{wk}, some known MOCSs in the literature
(see \cite{72-Tseng-Liu, 06-Appuswamy-Chaturvedi, 08-Wu-Spasojevi, 11-Han,18-Das, 18-Das-SPL, 19-Das} for example) can be directly used to construct sequences under SNCs.} It can be seen from Remark \ref{rem-han} that the systematic construction in \cite{11-Han} provides MOCSs with very flexible lengths.  Further, by Lemma \ref{thm_1}, if the sequences in the employed MOCSs have the same energy, the PAPR of the constructed CS is upper bounded by its set size.
\end{Remark}
}

\begin{Remark}
As pointed out by one of the anonymous reviewers, an idea of constructing sequences with zeros based on complementary sequences was reported recently in \cite{18-Ipanov}.
With a careful comparison, the main differences of our paper from  \cite{18-Ipanov} are the following.
\begin{itemize}
\item[1)] The motivation of our paper  is to construct sequences with low PAPR under SNCs for multicarrier communications, while the motivation of \cite{18-Ipanov} is to construct sequences with zero aperiodic correlation zone (ZACZ) for radars. This means that the main concern of our paper is different from  \cite{18-Ipanov}.
\item[2)]In our paper, we employ CSSs in MOCSs as seed sequences, and then insert zeros into two sequences taking from two different CSs to obtain sequences with low PAPR.  In \cite{18-Ipanov}, the authors used GCSs of a GCP as seed sequences, and then inserted zeros into two sequences from the GCP to get sequences with ZACZ.

\item[3)]The number of inserted zeros in our paper is arbitrary which does not depend on the length of the employed CSs, while the number of inserted zeros in  \cite{18-Ipanov} should be a multiple of the length of the employed GCPs.
\end{itemize}
\end{Remark}

To the best of our knowledge, the length of the MOCSs based on generalized Boolean functions in the literature is a power of two. This motivates us to construct MOCSs with non-power-of-two lengths based on generalized Boolean functions in the next section.

\section{MOCSs with non-power-of-two Lengths and Their CSs}

In this section, we introduce a direct construction of MOCSs based on generalized boolean functions, with its sequence length being non-power-of-two and PAPR upper bounded by 4. Before giving the construction of MOCSs, we need the following theorem.

\begin{Theorem}[Theorem 4 of \cite{16-ChenCY}]\label{CL}
Let $q$ be an even integer and  $L=2^{m-1}+2^v$,  where $m$ and $v$ are  integers with $m\geq 2$ and $1\leq v\leq m-1$. Define
\begin{eqnarray*}
 g_1(\mathbf{\underline{x}})&=& \frac{q}{2}\sum_{s=1}^{m-2}x_{\pi(s)}
  x_{\pi(s+1)}+\sum_{s=1}^{m-1}\lambda_{s}
  x_{\pi(s)}x_m+\sum_{s=1}^{m}\mu_{s}
  x_{s}+\mu, \\
  g_2(\mathbf{\underline{x}})&=&g_1(\mathbf{\underline{x}})
  +\frac{q}{2}x_{m}, \\
 g_3(\mathbf{\underline{x}})&=&g_1(\mathbf{\underline{x}})
 +\frac{q}{2}x_{\pi(1)},\\
  g_4(\mathbf{\underline{x}})&=&g_1(\mathbf{\underline{x}})
  +\frac{q}{2}x_{\pi(1)}+\frac{q}{2}x_{m},
\end{eqnarray*}
where $\mathbf{\underline{x}}\in \mathbb{Z}_{2}^{m},~\lambda_s,\mu,\mu_s$ are any given elements in $\mathbb{Z}_q$, and $\pi$ is a permutation of the symbols $\{1,2,\cdots,m-1\}$ with $\{\pi(1),\pi(2),\cdots,\pi(v)\}=\{1,2,\cdots,v\}$. Then $\{\psi(\mathbf{g}_1),\psi(\mathbf{g}_2),\psi(\mathbf{g}_3),
\psi(\mathbf{g}_4)\}$ is a CS of length $L$.
\end{Theorem}

\begin{Remark}
Note that $i_{\pi(s)}i_m=0$ for $0\leq i< L$ and $s>v$.
  $g_1(\mathbf{\underline{x}})$ in Theorem \ref{CL} can be reduced to
  \begin{eqnarray}\label{eq-g1}
  g_1(\mathbf{\underline{x}})&=& \frac{q}{2}\sum_{s=1}^{m-2}x_{\pi(s)}
  x_{\pi(s+1)}+\sum_{s=1}^{v}\lambda_{s}
  x_{\pi(s)}x_m+\sum_{s=1}^{m}\mu_{s}
  x_{s}+\mu,
  \end{eqnarray}
This observation will be used to calculate the code-rate of the CSSs
generated by Theorem \ref{CL}.
\end{Remark}

Now we state our main result in the following theorem, which is based on Theorem \ref{CL} and (\ref{eq-g1}).
\begin{Theorem}\label{thm 3}
In the context of Theorem \ref{CL} and $m\geq3,~1\leq v<m-1$, let
\begin{eqnarray*}
  f_{1}(\mathbf{\underline{x}}) &=& g_{1}(\mathbf{\underline{x}})+\frac{q}{2}
  x_{(m-1)}+\frac{q}{2}x_mx_{\pi(v)},\\
  f_{2}(\mathbf{\underline{x}}) &=& f_{1}(\mathbf{\underline{x}})+\frac{q}{2}x_{m},\\
   f_{3}(\mathbf{\underline{x}}) &=& f_{1}(\mathbf{\underline{x}})+
   \frac{q}{2}x_{\pi(1)},\\
   f_{4}(\mathbf{\underline{x}}) &=& f_{1}(\mathbf{\underline{x}})+\frac{q}{2}x_{\pi(1)}
   +\frac{q}{2}x_{m}.
\end{eqnarray*}
Then, the sets ${\mathcal{G}}=\{\psi(\mathbf{g}_1),\psi(\mathbf{g}_2),
\psi(\mathbf{g}_3),\psi(\mathbf{g}_4)\}$ and ${\mathcal{F}}=\{\psi(\mathbf{f}_1),\psi(\mathbf{f}_2),
\psi(\mathbf{f}_3),\psi(\mathbf{f}_4)\}$ form a $(2, 4, L)$-MOCS of sequence length $L=2^{m-1}+2^{v}$.
\end{Theorem}

\begin{proof}According to Theorem \ref{CL}, ${\mathcal{G}}$ and ${\mathcal{F}}$ are two CSs of sequence length $L$. Therefore, we only need to show that they are orthogonal.
For any integer $i$ and $u$, let $j=i+u$; also let $(i_{1},i_{2},\cdots,i_{m})$ and $(j_{1},j_{2},\cdots,j_{m})$ be the binary representations of $i$ and $j$, respectively. It is clear that
\begin{eqnarray*}
f_{2}(\underline{i})-g_{2}(\underline{j}) &=& f_{1}(\underline{i})-g_{1}(\underline{j})+
\frac{q}{2}(i_{m}-j_{m}), \\
f_{3}(\underline{i})-g_{3}(\underline{j}) &=& f_{1}(\underline{i})-g_{1}(\underline{j})+
\frac{q}{2}(i_{\pi(1)}-j_{\pi(1)}), \\
f_{4}(\underline{i})-g_{4}(\underline{j}) &=& f_{1}(\underline{i})-g_{1}(\underline{j})+
\frac{q}{2}(i_{m}-j_{m})+\frac{q}{2}(i_{\pi(1)}-j_{\pi(1)}).
\end{eqnarray*}
In order to show that ${\mathcal{G}}$ and ${\mathcal{F}}$ are orthogonal,  it is only necessary to prove
\begin{eqnarray}\label{eqn-condition}
\sum_{h=1}^{4}R_{\psi(\mathbf{f}_h),\psi(\mathbf{g}_h)}(u) &=&
\sum_{i=0}^{L-1-u}\xi^{f_1(\underline{i})-g_1(\underline{j})}
[1+(-1)^{i_{\pi(1)}- j_{\pi(1)}}][1+(-1)^{i_m-j_m}]
\end{eqnarray}
for $-L+1\leq u\leq L-1$.

\emph{Case 1: $u=0$}. In this case,  we have $i=j$ and (\ref{eqn-condition}) is reduced as
\begin{eqnarray*}
\sum_{h=1}^{4}R_{\psi(\mathbf{f}_h),\psi(\mathbf{g}_h)}(u)
&=&4\sum_{i=0}^{L-1}\xi^{f_1(\underline{i})-g_1(\underline{i})}\\
&=&4\sum_{i=0}^{L-1}\xi^{\frac{q}{2}i_{m-1}+\frac{q}{2}i_mi_{\pi(v)}}\\
&=&4\sum_{i=0}^{L-1}(-1)^{i_{m-1}+i_mi_{\pi(v)}}\\
&=&4\sum_{i=0}^{2^{m-1}-1}(-1)^{i_{m-1}+0\cdot i_{\pi(v)}}+4\sum_{i=2^{m-1}}^{2^{m-1}+2^{v}-1}(-1)^{0+1\cdot i_{\pi(v)}}\\
&=&4\sum_{i=0}^{2^{m-1}-1}(-1)^{i_{m-1}}+4\sum_{i=0}^{2^{v}-1}(-1)^{ i_{\pi(v)}}\\
&=&0
\end{eqnarray*}
where the last two equalities are due to  $\{\pi(1),\pi(2),\cdots,\pi(v)\}=\{1,2,\cdots,v\}$
and $v<m-1$.

\emph{Case 2:}  $0<u\leq L-1$. In this case,  the set $\{0,1,\cdots,L-1-u\}$ is divided into two disjoint subsets:
\begin{eqnarray*}
J_1(u)&=&\{0\le i\le L-1-u: i_{\pi(1)}\ne j_{\pi(1)}~ or~ j_m\ne i_m\}\\
J_2(u)&=&\{0\le i\le L-1-u: i_{\pi(1)}=j_{\pi(1)}, j_m= i_m\}.
\end{eqnarray*}
Hence (\ref{eqn-condition}) can be written as
\begin{eqnarray}\label{eqn-reduced}
\sum_{h=1}^{4}R_{\psi(\mathbf{f}_h),\psi(\mathbf{g}_h)}(u)&=&\sum_{i\in J_1(u)}\xi^{f_1(\underline{i})-g_1(\underline{j})}
[1+(-1)^{i_{\pi(1)}- j_{\pi(1)}}][1+(-1)^{i_m-j_m}]\nonumber\\
& &+\sum_{i\in J_2(u)}\xi^{f_1(\underline{i})-g_1(\underline{j})}
[1+(-1)^{i_{\pi(1)}- j_{\pi(1)}}][1+(-1)^{i_m-j_m}]\nonumber\\
&=&4\sum_{i\in J_2(u)}\xi^{f_1(\underline{i})-g_1(\underline{j})}.
\end{eqnarray}

For $i\in J_2(u)$, let $t$ be the smallest integer such that $i_{\pi(t)}\ne j_{\pi(t)}$, obviously, $t\ge2$. In particular, if $i_m= j_m=1$, which indicates that $2^{m-1}\leq i,j\leq2^{m-1}+2^v-1$, then we have $t\leq v$; Otherwise, since $\{\pi(1),\pi(2),\cdots,\pi(v)\}=
\{1,2,\cdots,v\}$ which implies $i_{s}=j_{s}$ for $s=1,2,\cdots,v$, then $j\geq i+2^v\geq 2^{m-1}+2^{v}$, and it contradicts the assumption.

Let $i'$ and $j'$ be integers which are different from $i$ and $j$ in only one position $\pi(t-1)$, i.e., $i'_{\pi(t-1)}=1-i_{\pi(t-1)}$ and $j'_{\pi(t-1)}=1-j_{\pi(t-1)}$ respectively, and so $j'=i'+u$ and $i'\in J_2(u)$. Then
\begin{eqnarray*}
f_{1}(\underline{i})-g_{1}(\underline{j})&=& \frac{q}{2}\sum_{s=1}^{m-2}(i_{\pi(s)}i_{\pi(s+1)}
-j_{\pi(s)}j_{\pi(s+1)})
+\sum_{s=1}^{v}\lambda_{s}(i_{\pi(s)}i_{m}-j_{\pi(s)}j_{m})\\
&&+\sum_{s=1}^{m}\mu_{s}(i_{s}-j_{s})+\frac{q}{2}i_{m-1}
+\frac{q}{2}i_mi_{\pi(v)},
\end{eqnarray*}
and
\begin{eqnarray*}
f_{1}(\underline{i}')-g_{1}(\underline{j}')&=& \frac{q}{2}\sum_{s=1}^{m-2}(i'_{\pi(s)}i'_{\pi(s+1)}
-j'_{\pi(s)}j'_{\pi(s+1)})
+\sum_{s=1}^{v}\lambda_{s}(i'_{\pi(s)}i'_{m}
-j'_{\pi(s)}j'_{m})\\
&&+\sum_{s=1}^{m}\mu_{s}(i'_{s}-j'_{s})+\frac{q}{2}i'_{m-1}
+\frac{q}{2}i'_mi'_{\pi(v)}.
\end{eqnarray*}
When $t=2$, it is clear that
\begin{eqnarray*}
f_{1}(\underline{i}')-g_{1}(\underline{j}')&=&
f_{1}(\underline{i})-g_{1}(\underline{j})
+\frac{q}{2}
(1-2i_{\pi(t-1)})i_{\pi(t)}-
\frac{q}{2}(1-2j_{\pi(t-1)})j_{\pi(t)}\\
&&+\lambda_{t-1}(1-2i_{\pi(t-1)})i_{m}-\lambda_{t-1}
(1-2j_{\pi(t-1)})j_{m}\\
&&+\mu_{t-1}(1-2i_{\pi(t-1)})-\mu_{t-1}(1-2j_{\pi(t-1)})\\
&\equiv &f_{1}(\underline{i})-g_{1}(\underline{j})+{q\over 2}~~(\bmod~q).\\
\end{eqnarray*}
On the other hand, when $t\ne2$ we have
\begin{eqnarray*}
f_{1}(\underline{i}')-g_{1}(\underline{j}')&=&f_{1}(\underline{i})-g_{1}(\underline{j})
+\frac{q}{2}i_{\pi(t-2)}(1-2i_{\pi(t-1)})+\frac{q}{2}
(1-2i_{\pi(t-1)})i_{\pi(t)}\\
&&-\frac{q}{2}j_{\pi(t-2)}(1-2j_{\pi(t-1)})-
\frac{q}{2}(1-2j_{\pi(t-1)})j_{\pi(t)}\\
&&+\lambda_{t-1}(1-2i_{\pi(t-1)})i_{m}-\lambda_{t-1}
(1-2j_{\pi(t-1)})j_{m}\\
&&+\mu_{t-1}(1-2i_{\pi(t-1)})-\mu_{t-1}(1-2j_{\pi(t-1)})\\
&\equiv &f_{1}(\underline{i})-g_{1}(\underline{j})+{q\over 2}~~(\bmod~q).
\end{eqnarray*}
This implies
\begin{eqnarray*}
\xi^{f_1(\underline{i})-g_1(\underline{j})}+
\xi^{f_1(\underline{i}')-g_1(\underline{j}')} = 0.
\end{eqnarray*}
When $i$ ranges over $J_2(u)$, so does $i'$. It then follows from (\ref{eqn-reduced}) that
$$
\sum_{h=1}^{4}R_{\psi(\mathbf{f}_h),\psi(\mathbf{g}_h)}(u)
=2\sum_{i\in J_2(u)}[\xi^{f_1(\underline{i})-g_1(\underline{j})}
+\xi^{f_1(\underline{i}')-g_1(\underline{j}')}]=0.
$$

\emph{Case 3:} $-L+1\leq u<0$. According to (\ref{eq-u}), for this case, we have
\begin{eqnarray}\label{eq-(-u)}
  \sum_{h=1}^{4}R_{\psi(\mathbf{g}_h),\psi(\mathbf{f}_h)}(\tau) =\sum_{i=0}^{L-1-\tau}\xi^{g_1(\underline{i})-f_1(\underline{j})}
[1+(-1)^{i_{\pi(1)}- j_{\pi(1)}}][1+(-1)^{i_m-j_m}] =0,
\end{eqnarray}
where $0<\tau\leq L-1$.

By similar arguments as in \emph{Case 2}, we can get (\ref{eq-(-u)}).

This completes the proof of this theorem.
\end{proof}

\begin{Corollary}\label{cor-bound}
With the same notations as in Theorem \ref{thm 3},
\begin{eqnarray*}
{\mathcal{W}}=\{\mathbf{w}_{1},\mathbf{w}_{2},
\mathbf{w}_{3},\mathbf{w}_{4}\}=
\{(\psi(\mathbf{f}_1)|\mathbf{0}_{b}|\psi(\mathbf{g}_1)),
(\psi(\mathbf{f}_2)|\mathbf{0}_{b}|\psi(\mathbf{g}_2)),
(\psi(\mathbf{f}_3)|\mathbf{0}_{b}|\psi(\mathbf{g}_3))
,(\psi(\mathbf{f}_4)|\mathbf{0}_{b}|\psi(\mathbf{g}_4))\}
\end{eqnarray*}
is a CS with sequence length $2^{m}+2^{v+1}+b$, and $\hbox{PAPR}(\mathcal{W})\leq 4$, where $(\psi(\mathbf{f}_i)|\mathbf{0}_{b}|\psi(\mathbf{g}_i))$ is constructed similarly as (\ref{eq-iter}), and $\mathbf{0}_{b}$ is a $b$-zeros sequence.
\end{Corollary}

\begin{proof}
First, by Theorems \ref{MO} and \ref{thm 3}, it can be observed that ${\mathcal{W}}$ is a CS of sequence length $2^{m}+2^{v+1}+b$. Second, for any $\mathbf{w}_i\in {\mathcal{W}}$, it contains $b$ zeros, which implies that the sequences in set ${\mathcal{W}}$ have the same energy. By Lemma \ref{thm_1},  $\hbox{PAPR}(\mathcal{W})\leq 4$.
\end{proof}

\begin{Example}
For $q=2,b=1$, $m=7$, let $\pi$ be the identity permutation of $\{1,2,\cdots,6\}$, and  $g_{1}(\underline{i})=\sum_{k=1}^5i_{k}i_{k+1}$. When $v=1,2,3,4,5$,
it is easy to check that ${\mathcal{G}}$ and ${\mathcal{F}}$ form an MOCS of length $L=66,68,72,80,96$. Then the set ${\mathcal{W}}$ constructed above is a CS of size $4$ with sequence length $133,137,145,161$, $193$, respectively. Table \ref{table-1} lists the PAPRs of these ${\mathcal{W}}$'s for various $v$. In addition,  Table \ref{table-2} lists the PAPRs of ${\mathcal{W}}$ for a fixed $L=66$ while the number of zeros is changed. It is worth noting that the PAPR of the proposed sequence sets is lower than the upper bound in Corollary \ref{cor-bound}. (Note that the oversampling ratio is 8 in our computation.)

\begin{table}[ht]
  \centering
    \caption{The PAPR of ${\mathcal{W}}$ in Example 1 for various $v$}\label{table-1}
\begin{tabular}{|c||c|c|c|c|c|}
  \hline
  Sequence length of $\mathcal{W}$ & 133 &137 & 145 & 161 & 193 \\\hline
  $\hbox{PAPR}(\mathcal{W})$ &  2.5986 & 2.9412 &  3.0956 & 3.6000 &  3.8036 \\
  \hline
\end{tabular}
\end{table}
\begin{table}[ht]\label{table-2}
  \centering
    \caption{The PAPR of $\mathcal{W}$ in Example 1 for various $b$ with $L=66$}\label{table-2}
\begin{tabular}{|c||c|c|c|c|c|}
  \hline
  $b$ & 0 &1 & 2 & 3 & 4 \\\hline
  $\hbox{PAPR}(\mathcal{W})$ & 2.4722 &  2.5986 & 2.5308 &   2.6484 &  2.5168 \\
  \hline
\end{tabular}
\end{table}
\end{Example}
{
\begin{Remark}
Compared with the MOCSs with flexible parameters from unitary-matrices in \cite{11-Han} (see Remark \ref{rem-han}), our construction in Theorem \ref{thm 3} based on generalized Boolean functions
only leads to non-power-of-two length MOCSs with fixed small size ($M=2$).  This means the iterative construction  in Section III to construct new CSs using the MOCSs in Theorem \ref{thm 3} will be ended after Step 2.
%
\end{Remark}
}

\section{The Code-rate and
minimum Hamming distance of proposed CSs}

In this section, first, we explicitly compute the set size of the codebook constructed in \cite{16-ChenCY} (see Theorem~\ref{CL}).  This result can then be immediately applied to calculate the set sizes and hence, the code-rates of our proposed construction also (see Corollary~\ref{cor-coderate}). Next, we provide the numerical values of the code-rates of our proposed codebooks (constructed from $\mathcal{G}$ and $\mathcal{F}$ in Theorem \ref{thm 3}) in a tabular form for various parameters. At last, the minimum Hamming distance of the codebook is discussed.

The following corollary gives the number of the codewords constructed by Theorem \ref{CL}.

\begin{Corollary}\label{thm-codebook}
The set size of the codebook $\mathcal{C}_1$ generated by Theorem \ref{CL} equals $v!(m-1-v)!q^{m+v+1}$.
\end{Corollary}
\begin{proof}
The conclusion  follows directly from Theorem \ref{CL} and (\ref{eq-g1}).
\end{proof}

\begin{Remark}\label{re-codebook}
It can be easily verified that when $b$ spectral nulls are considered as in our proposed codeword construction (see Corollary \ref{cor-bound}), the number of codewords is also equal to $v!(m-1-v)!q^{m+v+1}$.
\end{Remark}

The code-rate of our proposed CSs is then given by the following corollary which is a direct result obtained from Corollary \ref{thm-codebook} and Remark \ref{re-codebook}.

\begin{Corollary}\label{cor-coderate}
  In the context of Corollary \ref{cor-bound}, let $\mathcal{C}_2$ denote the codebook constructed by $\mathcal{W}$, then \begin{eqnarray*}
  R(\mathcal{C}_2)=\frac{\left\lfloor\log_q( v!\cdot(m-1-v)!\cdot q^{m+v+1})\right\rfloor}{2^m+2^{v+1}+b}.
\end{eqnarray*}
\end{Corollary}

According to Corollary \ref{cor-coderate}, we can obtain the following code-rate tables.
{\scriptsize
\begin{table*}[ht]
  \centering
\caption{$R(\mathcal{C}_2)$ for various $m,v,b$ when $q=2$}\label{table-code-rate}
\begin{tabular}{|c|c|c|c|c|c|c|c|c|c|c|}
  \hline
  \multirow{2}*{\diagbox[width=6em,trim=l]{$b$}{$R(\mathcal{C}_2)$}{$m,v$}} & $m=3$& \multicolumn{2}{|c|}{$m=4$} & \multicolumn{3}{|c|}{$m=5$} & \multicolumn{4}{|c|}{$m=6$} \\ \cline{2-11}
  & $v=1$& $v=1$ & $v=2$ & $v=1$ & $v=2$ &$v=3$ & $v=1$ & $v=2$ & $v=3$ &$v=4$ \\ \hline
  0&   0.4167 &  0.3500 & 0.3333 &  0.2500 &  0.2500& 0.2292 &  0.1765 &   0.1667 &  0.1625 &   0.1563 \\ \hline
  1&  0.3846 & 0.3333 & 0.3200 &   0.2432 & 0.2439& 0.2245 &  0.1739 &  0.1644 &  0.1605 & 0.1546 \\ \hline
  2&  0.3571 &  0.3182 &  0.3077 &   0.2368&  0.2381&  0.2200 &0.1714&  0.1622&  0.1585 & 0.1531 \\
  \hline
  3&  0.3333 &  0.3043 &  0.2963 &  0.2308 &  0.2326&  0.2157 &  0.1690 & 0.1600 &0.1566 &0.1515\\
  \hline
  4&  0.3125 &  0.2917& 0.2857&  0.2250& 0.2273&  0.2115 & 0.1667 &  0.1579 &0.1548 & 0.1500 \\
  \hline
  5& 0.2941 &  0.2800 & 0.2759& 0.2195&0.2222 & 0.2075 &  0.1644 & 0.1558& 0.1529& 0.1485 \\
  \hline
  6& 0.2778& 0.2692 &0.2667 &0.2143&0.2174 & 0.2037 &0.1622 &0.1538 &0.1512 &0.1471 \\
  \hline
\end{tabular}
\end{table*}
%
%
\begin{table*}[ht]
  \centering
  {
    \caption{$R(\mathcal{C}_2)$ for various $q,m,b$ when $v=1$}\label{table-code-rate q}
  \begin{tabular}{|c|c|c|c|c|c|c|c|c|c|}
  \hline
  \multirow{2}*{\diagbox[width=6em,trim=l]{$q$}{$R(\mathcal{C}_2)$}{$m,b$}} & \multicolumn{2}{|c|}{$m=4$}& \multicolumn{3}{|c|}{$m=5$} & \multicolumn{4}{|c|}{$m=6$} \\ \cline{2-10}
  & $b=1$& $b=2$&$b=1$ & $b=2$ &$b=3$ & $b=1$ & $b=2$ & $b=3$ &$b=4$ \\ \hline
  2&     0.3333 &   0.3182 & 0.2432 & 0.2368 &  0.2308& 0.1739 &  0.1714 &   0.1690 & 0.1667  \\ \hline
  4&    0.2857 & 0.2727 &0.2162 &0.2105 &  0.2051&   0.1449 &   0.1429 &  0.1408 & 0.1389 \\ \hline
  6&   0.2857 & 0.2727 & 0.2162 & 0.2105&   0.2051 &   0.1304 &  0.1286&  0.1268&   0.1250 \\
  \hline
  8&  0.2857& 0.2727 &0.1892 & 0.1842 &0.1795&   0.1304&   0.1286 & 0.1268 &   0.1250 \\
  \hline
  10&  0.2857 &0.2727&0.1892& 0.1842&  0.1795 &   0.1304 &  0.1286& 0.1268 &0.1250 \\
  \hline
\end{tabular}}

\end{table*}
}

Table \ref{table-code-rate} shows that for given $m$ and $v$, the code-rate $R(\mathcal{C}_2)$ decreases when $b$ increases (recall that $b$ is the number of zeros in the sequences). This is because of the  insertion of the zeros (or nulls) which increases the length of the sequences, but the number of codewords remains constant. However, it is worth noting that for large values of $m$, the code-rate does not decrease significantly with $b$. Similarly, we also observe that for a fixed value of $b$, the code-rate decreases  with $m$ and $v$.
Finally, Table \ref{table-code-rate q} shows that for given $m$, $b$ and fixed $v$ ($v=1$ in this case), the code-rate $R(\mathcal{C}_2)$  does not increase with $q$. This is expected because of
$$
R(\mathcal{C}_2)=\frac{(m+v+1)+\lfloor\log_q( v!\cdot(m-1-v)!)\rfloor}{2^m+2^{v+1}+b}.
$$

The minimum Hamming distance of $\mathcal{C}$ is identified by $$d_{min}(\mathcal{C})=min\{d\left(F_1,F_2\right) : F_1\ne F_2,F_1,F_2\in \mathcal{C}\},$$ where $d\left(F_1,F_2\right) = \omega\left(F_1-F_2\right)$, and $\omega\left(F_1-F_2\right)$ denotes the Hamming weight of the vector $F_1-F_2$.
$d_{min}(\mathcal{C}_2)$ might be small due to the small Hamming weight of $\mathbf{x}_m$ when the sequence length $L=2^{m-1}+2^v$ $(1\leq v\leq m-2)$. So we consider a subcode of $\mathcal{C}_2$.

\begin{Definition}
In the context of Theorems \ref{CL} and \ref{thm 3}, let $m\geq3,1\leq v<m-2$ and $b$ is a nonnegative integer, define
\begin{eqnarray*}
\mathcal{C}_3=\left\{\begin{array}{cc}
\left(\psi(\mathbf{f}_1)|\mathbf{0}_b|\psi(\mathbf{g}_1)\right):& \pi~\hbox{is a permutation of the symbols} ~\{1,2,\cdots,m-1\} \\
& \hbox{with}~\{\pi(1),\pi(2),\cdots,\pi(v)\}=\{1,2,\cdots,v\};\\
& \lambda_s=0~(1\leq s\leq v),\mu_m=0;\\
& \mu_s,\mu\in \mathbb{Z}_q~(1\leq s\leq m-1);
\end{array}
\right\}.
\end{eqnarray*}
It is obvious that $\mathcal{C}_3$ is a subcode of $\mathcal{C}_2$ and $|\mathcal{C}_3|=v!(m-1-v)!q^m$ with $$R(\mathcal{C}_3)=\frac{m+\lfloor\log_q( v!\cdot(m-1-v)!)\rfloor}{2^m+2^{v+1}+b}.$$
\end{Definition}

\begin{Lemma}\cite{00-Paterson}\label{lem-RM-Hammingdist}
For $q\geq 2$, the generalized $r$th-order Reed-Muller code, RM$_q(r,m)$, has minimum Hamming distance $2^{m-r}$.
  \end{Lemma}

\begin{Proposition}[Propositions 9 and 10 of \cite{16-ChenCY}]\label{lem-Hamming}
In the context of Theorem \ref{CL}, let $m\geq 3,1\leq v\leq m-2$, and
\begin{eqnarray*}
 Q_1=\frac{q}{2}\sum_{s=1}^{m-2}x_{\pi_1(s)}x_{\pi_1(s+1)},&&
 Q_2=\frac{q}{2}\sum_{s=1}^{m-2}x_{\pi_2(s)}x_{\pi_2(s+1)}
\end{eqnarray*}
where $\pi_1,\pi_2$ are any two permutations of the symbols $\{1,2,...,m-1\}$ with
$$\left\{\pi_s(1),\pi_s(2),\cdots,\pi_s(v)\}=\{1,2,\cdots,v\right\}~(s=1,2),$$ also let
\begin{eqnarray*}
\mathcal{S}_1=\left\{Q_1+\sum_{s=1}^{m-1}\mu_sx_s+\mu:
\mu_s,\mu\in\mathbb{Z}_q\right\},~\mathcal{S}_2=\left\{Q_2+\sum_{s=1}^{m-1}\mu_sx_s+\mu:
\mu_s,\mu\in\mathbb{Z}_q\right\},
\end{eqnarray*}
and
\begin{eqnarray*}
\mathcal{C}_{11}=\left\{\psi(\mathbf{f}):f(\underline{\mathbf{x}})\in\mathcal{S}_1\right\},
\mathcal{C}_{12}=\left\{\psi(\mathbf{f}):f(\underline{\mathbf{x}})\in\mathcal{S}_2\right\}.
\end{eqnarray*}
Then, for $\mathcal{\tilde{C}}_1=\mathcal{C}_{11}\bigcup\mathcal{C}_{12}$, we have
\begin{eqnarray*}
 d_{min}(\mathcal{\tilde{C}}_1) &=&\left\{\begin{array}{ll}
                                            2^{m-2}, &\hbox{if}~ \pi_1=\pi_2;  \\
                                            2^{m-3}, & \hbox{if}~ \pi_1\ne\pi_2.
                                          \end{array}\right.
\end{eqnarray*}

\end{Proposition}

Based on Proposition \ref{lem-Hamming} and Corollary \ref{cor-bound}, we can get the results about $d_{min}(\mathcal{C}_3)$ in the following corollary.

\begin{Corollary}\label{cor-Hamming-pi}
In the context of Corollary \ref{cor-bound} and Proposition \ref{lem-Hamming}, let
$\mathcal{C}_{21}=\left\{\left(\psi(\mathbf{f}_1)|\mathbf{0}_b|\psi(\mathbf{g}_1)\right):
f_1(\underline{\mathbf{x}})\in\mathcal{S}_{1}\right\}$ and $\mathcal{C}_{22}=\left\{(\psi(\mathbf{f}_2)|\mathbf{0}_b|\psi(\mathbf{g}_2)):
f_2(\underline{\mathbf{x}})\in\mathcal{S}_{2}\right\}$ generated in Corollary \ref{cor-bound}. Then, for $\mathcal{\tilde{C}}_2=\mathcal{C}_{21}\bigcup\mathcal{C}_{22}$, we have
\begin{eqnarray*}
 d_{min}(\mathcal{\tilde{C}}_2) &=&\left\{\begin{array}{ll}
                                            2^{m-1}, &\hbox{if}~ \pi_1=\pi_2~
                                            \hbox{and}~1\leq v<m-1;  \\
                                            2^{m-2}, & \hbox{if}~ \pi_1\ne\pi_2~
                                            \hbox{and}~1\leq v<m-2.
                                          \end{array}\right.
\end{eqnarray*}
\end{Corollary}
\begin{proof}
Case 1: If $\pi_1=\pi_2$ and $1\leq v<m-1$, then we have $\mathcal{C}_{21}=\mathcal{C}_{22}$. For any two $f_1(\underline{\mathbf{x}}),f_2(\underline{\mathbf{x}})\in\mathcal{C}_{11}$, it can be easily get that
      $f_1(\underline{i})-f_2(\underline{i})=
     g_1(\underline{i})-g_2(\underline{i})$ for any $0\leq i\leq 2^{m-1}+2^v-1$, and $$\omega(\psi(\mathbf{f}_1)-\psi(\mathbf{f}_2))=
     \omega\left(\psi(\mathbf{g}_1)-\psi(\mathbf{g}_2)\right),$$ so we have $$d\left((\psi(\mathbf{f}_1)|\mathbf{0}_b|\psi(\mathbf{g}_1)),
     (\psi(\mathbf{f}_2)|\mathbf{0}_b|\psi(\mathbf{g}_2)\right)=2\cdot\omega
     \left(\psi(\mathbf{f}_1)-\psi(\mathbf{f}_2)\right).$$ Hence, according to the first case in Proposition \ref{lem-Hamming}, it can be obtained that
     $d_{min}(\mathcal{\tilde{C}}_2)$ is equal to $2\cdot 2^{m-2}=2^{m-1}$.

 Case 2: If $\pi_1\ne\pi_2$ and $1\leq v<m-2$, let $$\mathcal{S}_{11}=\left\{g_1(\underline{\mathbf{x}})=f_1(\underline{\mathbf{x}})
      +\frac{q}{2}x_{\pi_1(v)}x_m+\frac{q}{2}x_{m-1}:f_1(\underline{\mathbf{x}})\in \mathcal{S}_1\right\},$$ $$\mathcal{S}_{12}=\left\{g_2(\underline{\mathbf{x}})=f_2(\underline{\mathbf{x}})
      +\frac{q}{2}x_{\pi_2(v)}x_m+\frac{q}{2}x_{m-1}:f_2(\underline{\mathbf{x}})\in \mathcal{S}_2\right\},$$
      also let $$\mathcal{C}'_{11}=\left\{\psi(\mathbf{g}_1):g_1(\underline{\mathbf{x}})\in\mathcal{S}_{11}\right\},
      ~\mathcal{C}'_{12}=\left\{\psi(\mathbf{g}_2):g_2(\underline{\mathbf{x}})\in\mathcal{S}_{12}\right\}.$$
      For any two codewords $\mathbf{c}_1\in\mathcal{C}'_{11}$ and $\mathbf{c}_2\in\mathcal{C}'_{12}$, let $\mathbf{d}=\mathbf{c}_1-\mathbf{c}_2$ and $\mathbf{d}^{(2^{m-1})}$ be the sequence of length $2^{m-1}$ by truncating $\mathbf{d}$. Since $\mathbf{d}^{(2^{m-1})}$ is a nonzero codeword of RM$_q(2,m-1)$, the Hamming weight of $\mathbf{d}$ is at least $2^{m-3}$ according to Lemma \ref{lem-RM-Hammingdist}. So, $d_{min}(\mathcal{C}'_{11}\bigcup\mathcal{C}'_{12})$ is larger than $2^{m-3}$, and then $d_{min}(\mathcal{\tilde{C}}_2)$ is larger than $2\cdot2^{m-3}=2^{m-2}$. Next, we exhibit two codewords in $\mathcal{\tilde{C}}_2$ where the distance is equal to $2^{m-2}$. Let
      \begin{eqnarray*}
      f_1(\mathbf{\underline{x}})&=&\frac{q}{2}\left(x_1x_2+x_2x_3+\cdots+x_{m-4}x_{m-3}
      +x_{m-3}x_{m-2}+x_{m-2}x_{m-1}\right), \\
      g_1(\mathbf{\underline{x}})&=&f_1(\mathbf{\underline{x}})+\frac{q}{2}x_vx_m+\frac{q}{2}x_{m-1}, \\
      f_2(\mathbf{\underline{x}})&=&\frac{q}{2}\left(x_1x_2+x_2x_3+\cdots+x_{m-4}x_{m-3}
      +x_{m-3}x_{m-1}+x_{m-1}x_{m-2}\right),\\
      g_2(\mathbf{\underline{x}})&=&f_2(\mathbf{\underline{x}})+\frac{q}{2}x_vx_m+\frac{q}{2}x_{m-1},
      \end{eqnarray*}
      also let $\mathbf{c}_1=\left(\psi(\mathbf{f}_1)|\mathbf{0}_b|\psi(\mathbf{g}_1)\right)$,
      $\mathbf{c}_2=\left(\psi(\mathbf{f}_2)|\mathbf{0}_b|\psi(\mathbf{g}_2)\right)$
      and $\mathbf{d}=\mathbf{c}_1-\mathbf{c}_2$.
      then we have
      \begin{eqnarray*}
        f_1(\mathbf{\underline{x}})-f_2(\mathbf{\underline{x}})=g_1(\mathbf{\underline{x}})
        -g_2(\mathbf{\underline{x}})=x_{m-3}x_{m-2}-x_{m-3}x_{m-1}.
      \end{eqnarray*}
      So $$\omega(\mathbf{d})=2\cdot\omega(\psi(\mathbf{f}_1)-\psi(\mathbf{f}_2))
      =2\cdot\omega\left(\mathbf{x}_{m-3}\mathbf{x}_{m-2}-\mathbf{x}_{m-3}\mathbf{x}_{m-1}\right)
      =2\cdot2^{m-3}=2^{m-2}.$$

This completes the proof of this corollary.
\end{proof}

The following is a straightforward result of Corollary \ref{cor-Hamming-pi}.

\begin{Corollary}\label{cor-Hamming}
In the context of Corollary \ref{cor-bound} and Proposition \ref{lem-Hamming}, we have
\begin{eqnarray*}
 d_{min}(\mathcal{\tilde{C}}_3) &=&\left\{\begin{array}{ll}
                                            2^{m-1}=4, &\hbox{if}~ m=3~
                                            \hbox{and}~1\leq v<m-1;  \\
                                            2^{m-2}, & \hbox{if}~ m>3~
                                            \hbox{and}~1\leq v<m-2.
                                          \end{array}\right.
\end{eqnarray*}
\end{Corollary}

\begin{proof}
  Case 1: For $m=3$ and $1\leq v<m-1$, there is only permutation $\pi$ of $\{1,...,m-1\}$ satisfying $1\leq v<m-1$ where $\pi(1)=1,\pi(2)=2$, then, according to Corollary \ref{cor-Hamming-pi}, we have
  $d_{min}(\mathcal{\tilde{C}}_3)=2^{m-1}$.

 Case 2: For $m>3$ and $1\leq v<m-2$, there exist different permutations of $\{1,...,m-1\}$ satisfying $1\leq v<m-2$, then by Corollary \ref{cor-Hamming-pi}, we have $d_{min}(\mathcal{\tilde{C}}_3)=min\{2^{m-1},2^{m-2}\}=2^{m-2}$.

This completes the proof of this corollary.
\end{proof}

\begin{Remark}
Since $\mathcal{C}_3$ is a subcode of $\mathcal{C}_2$, so we have $d_{min}(\mathcal{C}_2)\leq d_{min}(\mathcal{C}_3)$. $d_{min}(\mathcal{C}_2)$ would increase after deleting some codewords in $\mathcal{C}_2$.
For example, in the context of Corollary \ref{cor-Hamming-pi}, the codebook $\mathcal{C}_{21}$ is also a subcode of $\mathcal{C}_2$ with $d_{min}(\mathcal{C}_{21})=2^{m-1}$ for $m\geq 3$.
\end{Remark}

Table \ref{table-code-rate-distance} shows the code-rate and minimum Hamming distance of $\mathcal{C}_3$ for various $b,m,v$ when $q=2$, i.e., $\mathcal{C}_3$ is a binary code. It can be observed that $d_{min}(\mathcal{C}_3)$ is independent with $b$ and $v$ where $1\leq v<m-2$.
\begin{table*}[ht]
  \centering
    \caption{The code-rate and minimum Hamming distance of binary code $\mathcal{C}_3$ for various $b,m,v$}\label{table-code-rate-distance}
  \begin{tabular}{|c|c|c|c|c|c|}
  \hline
  $b$ & $m$& $v$ &length&code-rate&$d_{min}$ \\  \hline
  \multirow{7}{*}{0}& 3& 1&6 &  0.2500 & 4 \\ \cline{2-6}
  &   4 &  1 &10 &0.2500  &  4  \\ \cline{2-6}
  &  \multirow{2}{*}{5} & 1 &  18 & 0.1944 & 8   \\ \cline{3-6}
  && 2 & 20 &  0.1750  & 8\\ \cline{2-6}
  &  \multirow{3}{*}{6} & 1 &  34 & 0.1471 &  16  \\ \cline{3-6}
  && 2 & 36 &  0.1250   & 16 \\ \cline{3-6}
  && 3 & 40 &  0.1125  & 16 \\ \hline
    \multirow{7}{*}{1}& 3& 1&7 &0.2308  & 4 \\ \cline{2-6}
  &     4 &  1 &11 &0.2381  & 4\\ \cline{2-6}
  &  \multirow{3}{*}{5} & 1 &  19 & 0.1892 & 8   \\ \cline{3-6}
  && 2 & 21 &  0.1707  & 8\\ \cline{2-6}
 &  \multirow{3}{*}{6} & 1 &  35 &  0.1449 &  16  \\ \cline{3-6}
  && 2 & 37 & 0.1233    &16 \\ \cline{3-6}
  && 3 & 41 & 0.1111   &16 \\ \hline
    \multirow{7}{*}{2}& 3& 1&8 & 0.2143 &4  \\ \cline{2-6}
  &    4 &  1 &12 &  0.2273 & 4\\ \cline{2-6}
  &  \multirow{2}{*}{5} & 1 &  20&  0.1842 & 8   \\ \cline{3-6}
 && 2 & 22 &  0.1667  & 8\\ \cline{2-6}
&  \multirow{3}{*}{6} & 1 &  36 & 0.1429 & 16   \\ \cline{3-6}
  && 2 & 38 &  0.1216   & 16\\ \cline{3-6}
  && 3 & 42 & 0.1098   & 16\\ \hline
\end{tabular}
\end{table*}

\section{Concluding Remarks}
In this paper, we presented a systematic construction of CSs (CSs) with low PAPR and spectral nulls at certain positions. Our proposed sequences may find applications as low PAPR codewords for practical OFDM transmissions where certain sub-carriers are nulled or not allowed to transmit, for example, as transmission codewords for non-contiguous OFDM transmission in Cognitive Radio. As one of our main results, we proposed iterative techniques to generate MOCSs which can then be used to design more general CSs with spectral nulls not only at the center but also at other positions within the sequences. Using tools from generalized Boolean functions, we then provided a new construction of the \emph{seed} MOCSs with non-power-of-two lengths, which can be used to construct CSs. Our final proposed CSs have low PAPR, non-power-of-two lengths and can be used as OFDM codewords with any number of spectral nulls at the middle within the sequences. It would be interesting to construct MOCSs with large size and non-power-of-two lengths, and then obtain sequences with more flexible nulls under the proposed iterative construction.

\section*{Acknowledgments}
The authors are very grateful to the three anonymous reviewers and the Associate Editor, Prof.  Albert Guillen i Fabregas, for their valuable comments that improved the presentation and quality of this paper.  Special thanks go
to one of the reviewers for bringing our attention to a recent work in \cite{18-Ipanov}.


\bibliographystyle{IEEEtran}
\bibliography{IEEEfull.bib}
\end{document}